\documentclass[conference]{IEEEtran}
\IEEEoverridecommandlockouts
\usepackage{amsmath,amsfonts,amssymb}
\usepackage{amsthm}
\usepackage{graphicx}
\usepackage{algorithmicx}
\usepackage[ruled,linesnumbered]{algorithm2e}
\usepackage{algpseudocode}
\usepackage{textcomp}
\usepackage{makecell}
\usepackage{multirow,multicol}
\usepackage[colorlinks]{hyperref}
\hypersetup{citecolor=blue}
\usepackage[table]{xcolor}
\usepackage{threeparttable}
\usepackage{booktabs}
\usepackage{float}
\usepackage{balance}
\usepackage{tikz}
\usetikzlibrary{positioning,arrows.meta,calc,fit}
\hypersetup{hidelinks}
\newtheorem{proposition}{Proposition}

\newtheorem{remark}{Remark}
\newtheorem{corollary}{Corollary}

\usepackage[font=small,skip=2pt]{caption}

\usepackage{enumitem}
\setlist{nosep}

\begin{document}
\bstctlcite{IEEEtranBSTcontrol}
\IEEEaftertitletext{\vspace{-6mm}}
\title{CSI Reconstruction in Fluid Antenna Systems \\Without Spatial Covariance Priors\vspace{-6mm}}
\author{Zhentian Zhang, Kaitao Meng, Tuo Wu, Kai-Kit Wong, Hao Xu, Liang Liu, Pei Xiao, Chao Wang, Kin-Fai Tong
	
	\thanks{ }
	\thanks{Zhentian Zhang is with the Department of Electrical and Electronic Engineering, The Hong Kong Polytechnic University, Hong Kong, SAR (e-mail: zhentianzhangzzt@gmail.com).}
	\thanks{Kaitao Meng is with the Department of Electrical and Electronic Engineering, University of Manchester, Manchester, UK (email: kaitao.meng@manchester.ac.uk).}
	\thanks{Tuo Wu is with the School of Electronic and Information Engineering, South China University of Technology, Guangzhou 510640, China (E-mail: wutuo@scut.edu.cn).}
	\thanks{Kai-Kit Wong is with the Department of Electronic and Electrical Engineering, University College London, Torrington Place, WC1E 7JE, United Kingdom  (e-mail: kai-kit.wong@ucl.ac.uk).}
	\thanks{Hao Xu is with the National Mobile Communications Research Laboratory, Frontiers Science Center for Mobile Information Communication and Security, Southeast University, Nanjing, 210096, China (e-mail: hao.xu@seu.edu.cn).}
	\thanks{Liang Liu is with the Department of Electrical and Electronic Engineering, The Hong Kong Polytechnic University, Hong Kong, SAR (liang-eie.liu@polyu.edu.hk).}
	\thanks{Pei Xiao is with the 5G \& 6G Innovation Centre, University of Surrey, U. K. (email: p.xiao@surrey.ac.uk).}
	\thanks{Chao Wang is with the Integrated Service Networks Laboratory, Xidian University, Xi’an 710071, China (e-mail: drchaowang@126.com).}
	\thanks{Kin-Fai Tong is with the School of Science and Technology, Hong Kong Metropolitan University, Hong Kong, China (E-mail: ktong@hkmu.edu.hk)}
}

\maketitle

\begin{abstract}
	Fluid antenna systems (FASs) exploit many candidate ports for spatial diversity, but hardware constraints allow channel observations at only a few active ports. Whether full-port CSI can be recovered without pre-acquired channel statistics remains open. Under the Clarke isotropic scattering model, we show that the channel lies in a low-dimensional spatial modal subspace determined by the scattering environment rather than the total port count. Consequently, recovery becomes feasible when the number of observed ports reaches the modal dimension (i.e., $M\geq r$), even when $M\ll N$. We further establish a sharp feasibility threshold: reliable recovery is impossible below this dimension regardless of SNR, whereas accuracy improves with additional observations above it. By decomposing the recovery error into modal truncation, estimation, and learning components, we derive explicit tradeoffs among RF chains, pilot overhead, transmit power, and training data. These results enable scalable prior-free full-port CSI recovery with few active ports.
\end{abstract}


\section{Introduction}\label{sec:intro}

Fluid antenna systems (FASs)~\cite{fas-twc-21,kit_electronic} exploit a long-overlooked physical phenomenon: \emph{channel fading can vary substantially over sub-wavelength spatial offsets}. This fine-grained variation, accurately characterized by the classical Jakes'~\cite{Jakes1} and Clarke's~\cite{Clarke2} models, arises from the spatially correlated fading field observed by densely packed antenna ports. Through rapid port switching, FASs exploit both fading randomness \cite{r0,r1,r2,r3,r4,r5,r6,r7} and reconfigurable geometry~\cite{FAA1,FAA2,XJY,LHY,CKJ}, thereby achieving substantial performance gains with limited hardware cost.

\subsection{Related Work and Motivation}

Although FAS performance analyses often assume channel state information (CSI), CSI must be acquired before these gains can be realized, particularly in multi-user systems. A central challenge is to reconstruct the full channel over $N$ candidate ports from observations at only $M\ll N$ active ports. The fundamental question is whether full-port CSI can be recovered from such sparse observations without \emph{any} statistical prior and, if so, under what conditions and resource requirements.

Existing approaches rely on two types of prior knowledge. The first assumes access to the spatial covariance matrix $\boldsymbol{\Sigma}$. Under this assumption, \cite{Ideal_bound} derives the fundamental MMSE normalized mean-square error (NMSE) bound and shows that accurate reconstruction is possible under arbitrary port-selection strategies. Subsequent works learn $\boldsymbol{\Sigma}$ offline from training data and use it for online Bayesian reconstruction~\cite{CSI3}. However, because covariance depends on the propagation environment, reliable offline learning can be difficult in practice.

The second approach adopts geometry-structured channel models~\cite{Geometry0}, in which a finite number of scatterers parameterize propagation and implicitly determine $\boldsymbol{\Sigma}$. This low-dimensional representation enables compact and robust inference~\cite{Geometry1,Geometry3,Geometry4,Geometry5}. Notably, \cite{Geometry0} achieves near-optimal performance with linear complexity. However, finite-scatterer models may be invalid in rich-scattering environments. 

{\em It therefore remains unclear whether full-port reconstruction is possible without either covariance knowledge or a prescribed propagation geometry structure. This prior-free setting motivates this work.}

\subsection{Contributions}

This work establishes the feasibility limits and resource requirements of prior-free full-port CSI reconstruction for FASs. The main contributions are as follows.

\begin{itemize}
	\item \textbf{Modal Domain Representation:}
	Under the Clarke isotropic scattering model~\cite{Clarke2}, we show that an $N$-port FAS channel is effectively represented by a finite set of dominant spatial modes. The modal dimension is determined by the scattering environment rather than the total port count. This structure makes sparse-port reconstruction feasible once the observed ports resolve the modal subspace (i.e., $M\geq r$), even when $M\ll N$. Moreover, the modal truncation error decreases as the number of candidate ports increases.
	
	\item \textbf{Theoretical Performance Interpretation:}
	We derive a hierarchy of NMSE bounds and establish a sharp threshold governed by the modal dimension. Below this threshold, reliable full-port recovery is information-theoretically impossible regardless of SNR or prior knowledge. Above it, recovery accuracy improves as more ports are observed. We further decompose the total error into modal truncation, estimation, and learning components, capturing the effects of candidate-port count, noisy observations, and finite training data, respectively. 
\end{itemize}

The remainder of this paper is organized as follows. Section~\ref{sec:problem_formulation} introduces the sparse observation model and modal-domain channel representation. Section~\ref{sec:performance_bounds} derives the error decomposition, oracle bounds, and feasibility threshold. Section~\ref{sec:numerical_results} presents the numerical results. Finally, Section~\ref{sec:conclusion} concludes the paper.
\section{Modal Domain Channel Representation}\label{sec:problem_formulation}

\subsection{Channel Reconstruction from Sparse Port Observations}
\label{subsec:problem_formulation}

Consider a fluid antenna with aperture length $W\lambda$, where $W$ is the normalized aperture size and $\lambda$ is the carrier wavelength. Uniformly spaced candidate ports have normalized positions
\begin{equation}
	x_n = \frac{(n-1)\,W}{N-1},
	\label{eq:port_position}
\end{equation}
for $n=1,\dots,N$. At pilot snapshot $t$, the full-port channel is
\begin{equation}
	\mathbf g_t
	= \bigl[g_t(1),\, g_t(2),\, \dots,\, g_t(N)\bigr]^{\mathsf T}
	\in \mathbb C^{N}.
	\label{eq:full_channel_vector}
\end{equation}

Hardware constraints allow only a subset of ports to be active. Let $\mathcal O_t \subseteq \{1,\dots,N\}$ denote the active-port set, where $|\mathcal O_t|=M_t \ll N$, and let $\mathcal U_t = \{1,\dots,N\}\setminus\mathcal O_t$ denote the inactive ports. The selection matrix $\mathbf S_t \in \{0,1\}^{M_t\times N}$ extracts the entries indexed by $\mathcal O_t$. The noisy pilot observation is
\begin{equation}
	\mathbf y_t
	= \mathbf S_t \mathbf g_t + \mathbf v_t,
	\label{eq:direct_sampling_model}
\end{equation}
where $\mathbf v_t \sim \mathcal{CN}(\mathbf 0,\,\sigma^2\mathbf I_{M_t})$ is additive white Gaussian noise with variance $\sigma^2$.

Our objective is to recover $\mathbf g_t$ from $\mathbf y_t$, including channel estimation at $\mathcal O_t$ and spatial interpolation at $\mathcal U_t$. With known spatial correlation, both are jointly solved by the MMSE estimator \cite{Ideal_bound}. In practice, however, the correlation structure is unknown and must be learned jointly with the channels.

We consider $T$ pilot snapshots $\{\mathbf y_t\}_{t=1}^{T}$ sharing a structured prior parameterized by $\boldsymbol\theta$. The single-snapshot setting follows by setting $T=1$. An empirical-Bayes approach first estimates the prior from the marginal likelihood,
\begin{equation}
	\widehat{\boldsymbol\theta}
	= \arg\max_{\boldsymbol\theta}\;
	p(\mathbf Y;\boldsymbol\theta),
	\label{eq:empirical_bayes_theta}
\end{equation}
and then reconstructs each channel through posterior inference,
\begin{equation}
	\widehat{\mathbf g}_t
	= \mathbb E\!\bigl[\mathbf g_t \mid \mathbf y_t;
	\widehat{\boldsymbol\theta}\bigr],
	\label{eq:empirical_bayes_ghat}
\end{equation}
for $t=1,\dots,T$, where $\mathbf Y = \{\mathbf y_t\}_{t=1}^{T}$. Thus, estimation and interpolation are two components of the same full-port posterior reconstruction.

\subsection{Reduced-Rank Spatial Prior from Clarke's Model}
\label{subsec:clarke_prior}

We construct the prior from the Clarke isotropic scattering model \cite{r7,BC}. For uniformly spaced ports, the spatial covariance is
\begin{equation}
	\mathbf \Sigma(\vartheta)
	=
	\begin{bmatrix}
		a_{\vartheta}(0) & a_{\vartheta}(1) & \cdots & a_{\vartheta}(N-1)\\
		a_{\vartheta}(-1) & a_{\vartheta}(0) & \cdots & a_{\vartheta}(N-2)\\
		\vdots & \vdots & \ddots & \vdots\\
		a_{\vartheta}(-(N-1)) & a_{\vartheta}(-(N-2)) & \cdots & a_{\vartheta}(0)
	\end{bmatrix},
	\label{eq:toeplitz_covariance}
\end{equation}
with entries
\begin{equation}
	a_\vartheta(\ell)
	= \operatorname{sinc}(2\pi d_\vartheta\,\ell),
	\label{eq:sinc_entry}
\end{equation}
where $d_\vartheta \triangleq \vartheta/(N-1)$ and $\operatorname{sinc}(x) \triangleq \sin(x)/x$. The effective normalized aperture $\vartheta$ equals $W$ under the ideal Clarke model. We later learn $\vartheta$ to accommodate model mismatch.

The sequence $a_\vartheta(\ell)$ is generated by a rectangular spatial power spectrum,
\begin{equation}
	a_\vartheta(\ell)
	= \frac{1}{2\pi}
	\int_{-\pi}^{\pi}
	f_\vartheta(\omega)\,e^{j\omega\ell}\,d\omega,
	\label{eq:fourier_coefficient_relation}
\end{equation}
where
\begin{equation}
	f_\vartheta(\omega)
	= \frac{1}{2d_\vartheta}\,
	\mathbf 1\!\bigl\{|\omega| \le 2\pi d_\vartheta\bigr\},
	\label{eq:rectangular_spatial_spectrum}
\end{equation}
and $\mathbf 1\{\cdot\}$ denotes the indicator function. Direct integration gives
\begin{equation}
	\frac{1}{2\pi}
	\int_{-2\pi d_\vartheta}^{2\pi d_\vartheta}
	\frac{e^{j\omega\ell}}{2d_\vartheta}\,d\omega
	= \frac{\sin(2\pi d_\vartheta\,\ell)}{2\pi d_\vartheta\,\ell}
	= \operatorname{sinc}(2\pi d_\vartheta\,\ell).
	\label{eq:rectangular_spectrum_proof}
\end{equation}

Since $f_\vartheta(\omega)$ is band-limited, $\mathbf\Sigma(\vartheta)$ is effectively low rank. Let $\mathbf F_N \in \mathbb C^{N\times N}$ be the unitary DFT matrix with frequencies $\omega_k = 2\pi k/N$, $k=0,\dots,N{-}1$. By Szeg\H{o}'s theorem \cite{Gray_Toeplitz}, the DFT asymptotically diagonalizes the Toeplitz covariance,
\begin{equation}
	\mathbf\Sigma(\vartheta)
	\approx
	\mathbf F_N\,
	\operatorname{diag}\!\bigl(
	f_\vartheta(\omega_0),\dots,f_\vartheta(\omega_{N-1})
	\bigr)\,
	\mathbf F_N^{\mathsf H}.
	\label{eq:dft_diagonalization}
\end{equation}
Define the active DFT-mode set
\begin{equation}
	\mathcal K(\vartheta)
	\triangleq
	\bigl\{k : |\omega_k| \le 2\pi d_\vartheta\bigr\},
	\label{eq:active_dft_support}
\end{equation}
and let $r = |\mathcal K(\vartheta)|$. Its size depends on aperture rather than port count,
\begin{equation}
	r \approx 2\vartheta + 1 \ll N.
	\label{eq:effective_rank}
\end{equation}

Accordingly, let
\begin{equation}
	\mathbf B(\vartheta)
	\triangleq \mathbf F_{N,\mathcal K(\vartheta)}
	\in \mathbb C^{N\times r}
	\label{eq:reduced_basis}
\end{equation}
contain the active DFT columns, and model the channel as
\begin{equation}
	\mathbf g_t = \mathbf B(\vartheta)\,\mathbf z_t,
	\label{eq:latent_modal_model}
\end{equation}
where $\mathbf z_t \in \mathbb C^{r}$ follows
\begin{equation}
	\mathbf z_t \sim \mathcal{CN}(\mathbf 0,\,\mathbf P),
	\label{eq:zt_prior}
\end{equation}
with $\mathbf P = \operatorname{diag}(p_1,\dots,p_r) \succeq \mathbf 0$. Under the ideal Clarke model, the modal powers are approximately uniform over the active band. Here, $\mathbf P$ remains diagonal but unknown, preserving the low-dimensional structure while adapting to the channel statistics. The resulting covariance is
\begin{equation}
	\mathbf\Sigma(\boldsymbol\theta)
	= \mathbf B(\vartheta)\,\mathbf P\,\mathbf B(\vartheta)^{\mathsf H},
	\label{eq:reduced_covariance}
\end{equation}
which provides a learned rank-$r$ approximation to the Clarke covariance. The subsequent inference framework also applies to other orthonormal bases, including the exact eigenbasis. 

We illustrate the finite modal effect in Fig.~\ref{fig:Modal_Sparsity}. Since the Clarke spatial spectrum
$f_\vartheta(\omega)$ is strictly band-limited to
$|\omega|\le 2\pi d_\vartheta$, only $r\approx 2\vartheta+1$ of the $N$
DFT frequencies fall within its support. By Szeg\H{o}'s
theorem~\cite{Gray_Toeplitz}, the eigenvalues of the Toeplitz covariance
$\mathbf\Sigma(\vartheta)$ converge to samples of $f_\vartheta$, so
exactly $r$ eigenvalues are significant while the remaining $N-r$ vanish
asymptotically. The effective rank $r$ is therefore set by the
\emph{physical aperture} $\vartheta$ and is \emph{independent} of the
port count $N$, a separation corroborated by
Fig.~\ref{fig:Modal_Sparsity} across all tested configurations. This
decoupling is the key enabler of the reduced-rank model: the latent
dimension $r$ remains small even when $N$ is very large.
\begin{figure}[t!]
	\centering
	\includegraphics[width=\columnwidth]{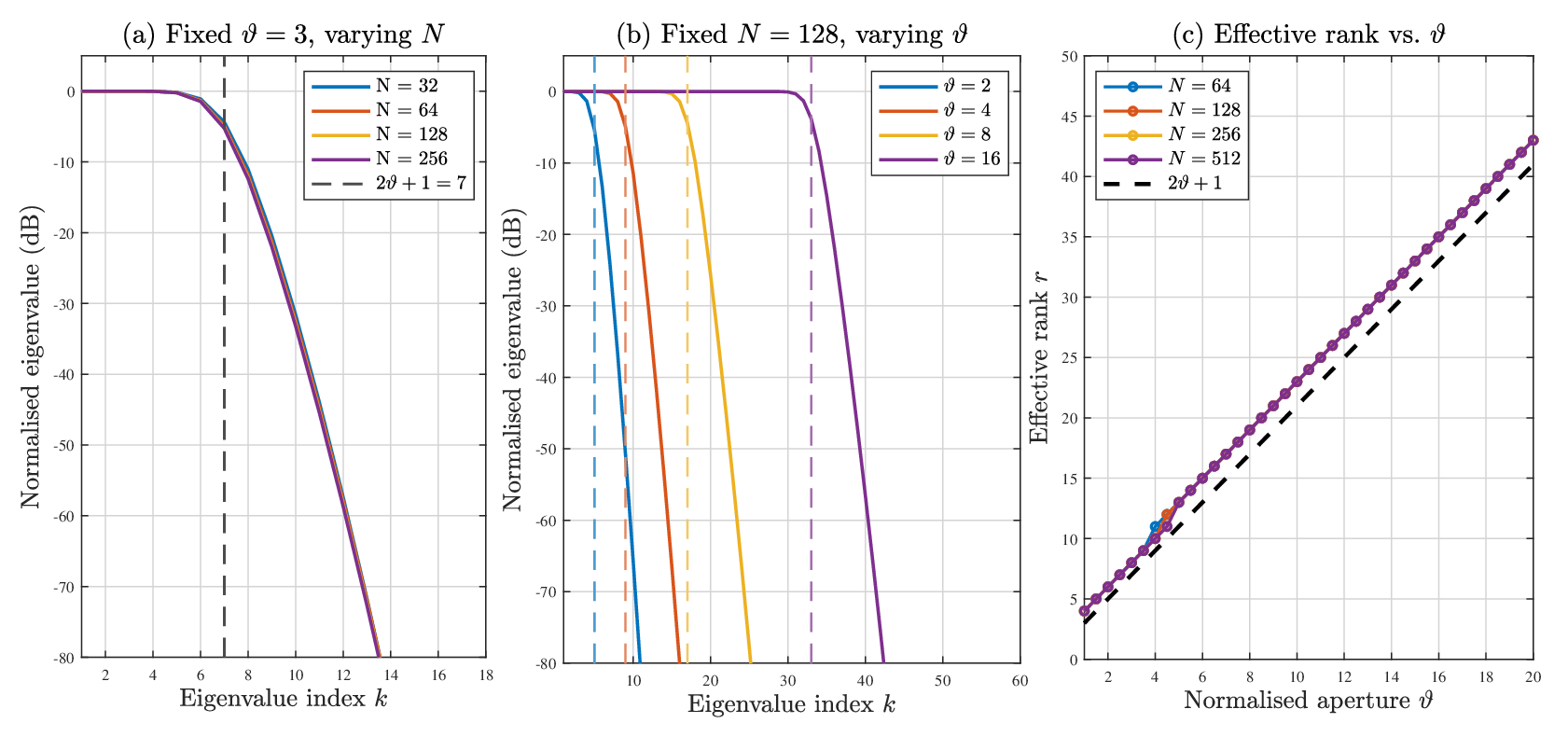}
	\caption{Eigenvalue spectra and effective ranks of the Clarke covariance $\mathbf{\Sigma}(\vartheta)$. (a) Normalized eigenvalues for $\vartheta=3$ and $N\in\{32,64,128,256\}$. (b) Normalized eigenvalues for $N=128$ and $\vartheta\in\{2,4,8,16\}$. Dashed lines indicate $2\vartheta+1$. (c) Effective rank, defined by eigenvalues above $-20$ dB relative to the maximum, versus $\vartheta$ for $N\in\{64,128,256,512\}$.}
	\label{fig:Modal_Sparsity}
\end{figure}

\subsection{Generalized Observation Model in the Modal Domain}
\label{subsec:signal_model}

Direct port sampling in \eqref{eq:direct_sampling_model} is a special case of a more general receiver model. Let $\mathbf\Phi_t \in \mathbb C^{L_t \times M_t}$ denote a linear front-end operator and define
\begin{equation}
	\mathbf\Psi_t
	\triangleq \mathbf\Phi_t\,\mathbf S_t
	\in \mathbb C^{L_t \times N}.
	\label{eq:overall_sensing_operator}
\end{equation}
Here, $L_t$ is the number of measurements at snapshot $t$. Without front-end processing, $\mathbf\Phi_t = \mathbf I_{M_t}$, such that $L_t = M_t$ and $\mathbf\Psi_t = \mathbf S_t$.

Substituting \eqref{eq:latent_modal_model} into the observation model gives
\begin{equation}
	\mathbf y_t
	= \mathbf\Psi_t\,\mathbf B(\vartheta)\,\mathbf z_t + \mathbf n_t
	= \mathbf A_t(\vartheta)\,\mathbf z_t + \mathbf n_t,
	\label{eq:modal_observation}
\end{equation}
where
\begin{equation}
	\mathbf A_t(\vartheta)
	\triangleq \mathbf\Psi_t\,\mathbf B(\vartheta)
	\in \mathbb C^{L_t \times r}
	\label{eq:effective_measurement_matrix}
\end{equation}
is the modal-domain sensing matrix, and $\mathbf n_t \sim \mathcal{CN}(\mathbf 0,\,\sigma^2\mathbf I_{L_t})$. The original $N$-dimensional reconstruction problem is therefore reduced to an $r$-dimensional linear system, where $r \approx 2W{+}1$ is typically small. We define the signal-to-noise ratio as $\mathrm{SNR} \triangleq \mathrm{tr}(\boldsymbol{\Sigma})/(N\sigma^2)$, namely the average per-port signal power divided by the noise variance.

\section{Performance Bounds}
\label{sec:performance_bounds}

This section characterizes the achievable reconstruction accuracy under
the reduced-rank model and the resources required to approach it. We
decompose the MSE into three resource-dependent terms, derive oracle
bounds, establish a sharp transition at $M = r$.

\subsection{Three-Term MSE Decomposition}
\label{subsec:mse_decomposition}

Let $\boldsymbol\Sigma_{\mathrm{true}}$ denote the true channel
covariance with eigenvalues
$\lambda_1 \ge \lambda_2 \ge \cdots \ge \lambda_N \ge 0$, and let
$\boldsymbol\Pi_{\!B} = \mathbf B\mathbf B^{\mathsf H}$ project onto
the reduced basis
$\mathbf B(\vartheta) \in \mathbb C^{N \times r}$.

\begin{proposition}[MSE decomposition]
	\label{prop:mse_decomposition}
	The full-port MSE of
	$\widehat{\mathbf g}_t = \mathbf B(\widehat\vartheta)\,
	\widehat{\mathbf z}_t$ satisfies
	\begin{equation}
		\mathrm{MSE}
		\triangleq
		\mathbb E\bigl[\|\mathbf g_t - \widehat{\mathbf g}_t\|^2\bigr]
		=
		\underbrace{\epsilon_{\mathrm{sub}}}_{\text{truncation}}
		+ \underbrace{\epsilon_{\mathrm{est}}}_{\text{estimation}}
		+ \underbrace{\epsilon_{\mathrm{learn}}}_{\text{learning}},
		\label{eq:mse_three_term}
	\end{equation}
	where
	\begin{align}
		\epsilon_{\mathrm{sub}}
		&= \operatorname{tr}\!\bigl(
		(\mathbf I_N - \boldsymbol\Pi_{\!B})\,
		\boldsymbol\Sigma_{\mathrm{true}}\,
		(\mathbf I_N - \boldsymbol\Pi_{\!B})
		\bigr),
		\label{eq:eps_sub}
		\\
		\epsilon_{\mathrm{est}}
		&= \operatorname{tr}\!\biggl(
		\Bigl(
		\mathbf P_{\mathrm{true}}^{-1}
		+ \frac{1}{\sigma^2}\,
		\mathbf G^{\mathsf H}\mathbf G
		\Bigr)^{\!-1}
		\biggr),
		\label{eq:eps_est}
		\\
		\epsilon_{\mathrm{learn}}
		&= \mathrm{MSE}(\widehat{\mathbf P},\widehat\sigma^2)
		- \mathrm{MSE}(\mathbf P_{\mathrm{true}},\sigma^2_{\mathrm{true}})
		\ge 0.
		\label{eq:eps_learn}
	\end{align}
	Here $\mathbf G = \mathbf S_{\mathcal O}\,\mathbf B(\vartheta)
	\in \mathbb C^{M \times r}$ is the effective sensing matrix,
	$\mathbf P_{\mathrm{true}}
	= \mathbf B^{\mathsf H}\boldsymbol\Sigma_{\mathrm{true}}\mathbf B$
	is the oracle modal covariance, and
	$\epsilon_{\mathrm{learn}}$ is the excess MSE caused by learned
	rather than oracle parameters.
\end{proposition}

\begin{IEEEproof}
	Write
	$\mathbf g_t = \boldsymbol\Pi_{\!B}\mathbf g_t
	+ (\mathbf I - \boldsymbol\Pi_{\!B})\mathbf g_t
	\triangleq \mathbf g_t^{\parallel} + \mathbf g_t^{\perp}$.
	Since $\widehat{\mathbf g}_t \in \operatorname{col}(\mathbf B)$,
	the error decomposes orthogonally as
	\begin{equation}
		\|\mathbf g_t - \widehat{\mathbf g}_t\|^2
		= \|\mathbf g_t^{\perp}\|^2
		+ \|\mathbf g_t^{\parallel} - \widehat{\mathbf g}_t\|^2.
		\label{eq:orthogonal_decomp}
	\end{equation}
	Taking expectations yields
	$\epsilon_{\mathrm{sub}} = \mathbb E[\|\mathbf g_t^{\perp}\|^2]$.
	Let
	$\mathbf g_t^{\parallel} = \mathbf B\mathbf z_t^{\mathrm{true}}$,
	where
	$\mathbf z_t^{\mathrm{true}} = \mathbf B^{\mathsf H}\mathbf g_t$.
	The observation becomes
	$\mathbf y_t = \mathbf G\mathbf z_t^{\mathrm{true}}
	+ \mathbf S_{\mathcal O}\mathbf g_t^{\perp} + \mathbf n_t$.
	When $\epsilon_{\mathrm{sub}} \ll \sigma^2$, the residual term is
	absorbed into the noise floor. Under oracle modal parameters, the
	LMMSE covariance is
	$(\mathbf P_{\mathrm{true}}^{-1}
	+ \sigma^{-2}\mathbf G^{\mathsf H}\mathbf G)^{-1}$, whose trace
	gives $\epsilon_{\mathrm{est}}$. Since oracle parameters minimize the
	Bayes risk, the excess error $\epsilon_{\mathrm{learn}}$ is
	nonnegative.
\end{IEEEproof}

The three terms are controlled by distinct resources.
\begin{itemize}
	\item $\epsilon_{\mathrm{sub}}$ depends on the physical aperture $W$
	and modal dimension $r$.
	\item $\epsilon_{\mathrm{est}}$ depends on active-port selection,
	the number $M$, and SNR.
	\item $\epsilon_{\mathrm{learn}}$ depends on the available learning
	data $(T, M, \mathrm{SNR})$.
\end{itemize}
FAS reconstruction therefore involves a three-way tradeoff among
aperture, spatial observations, and training data.

\subsection{Subspace Truncation Error}
\label{subsec:subspace_truncation}

The truncation term measures the channel energy outside the reduced
modal subspace.

\begin{proposition}[Truncation error under Clarke-type correlation]
	\label{prop:truncation_error}
	Let $\boldsymbol\Sigma_{\mathrm{true}}$ be the Clarke covariance
	with aperture $W$, and let
	$\mathbf B(\vartheta) \in \mathbb C^{N \times r}$ collect the
	$r = |\mathcal K(\vartheta)|$ active DFT columns. Then
	\begin{enumerate}
		\item[\emph{(i)}]
		The truncation error is the spectral energy outside the modeled
		band,
		\begin{equation}
			\epsilon_{\mathrm{sub}}
			= \sum_{k \notin \mathcal K(\vartheta)}
			f_W(\omega_k)
			\approx \frac{N}{2\pi}
			\int_{|\omega| > 2\pi d_\vartheta}
			f_W(\omega)\,d\omega.
			\label{eq:eps_sub_spectral}
		\end{equation}
		
		\item[\emph{(ii)}]
		If $\vartheta \ge W$, then for large $N$,
		\begin{equation}
			\frac{\epsilon_{\mathrm{sub}}}
			{\operatorname{tr}(\boldsymbol\Sigma_{\mathrm{true}})}
			\le \frac{C}{N},
			\label{eq:truncation_rate}
		\end{equation}
		where $C > 0$ depends only on $W$. Hence
		$\epsilon_{\mathrm{sub}} \to 0$ as $N \to \infty$.
	\end{enumerate}
\end{proposition}

\begin{IEEEproof}
	By Szeg\H{o}'s theorem, the diagonal entries of
	$\mathbf F_N^{\mathsf H}\boldsymbol\Sigma_{\mathrm{true}}\mathbf F_N$
	converge to samples of $f_W(\omega)$ on the DFT grid. Summing the
	energy outside $\mathcal K(\vartheta)$ gives
	\eqref{eq:eps_sub_spectral}. For $\vartheta \ge W$, the rectangular
	spectrum $f_W(\omega)$ vanishes outside the modeled band. The
	remaining DFT approximation error arises from band-edge leakage and
	decays as $\mathcal O(1/N)$.
\end{IEEEproof}

Thus, when $r \ge 2W{+}1$ and $N$ is sufficiently large, truncation is
negligible and the dominant errors are $\epsilon_{\mathrm{est}}$ and
$\epsilon_{\mathrm{learn}}$.

\begin{remark}[Effective rank and truncation error]
	The relation $r=2\lfloor W\rfloor+1$ specifies the effective modal
	dimension of the Clarke covariance, not necessarily its exact
	finite-$N$ rank. If the true covariance has rank $r$ and range
	contained in $\operatorname{col}(\mathbf B)$, then
	$\epsilon_{\mathrm{sub}}=0$. Otherwise, finite-dimensional leakage
	and basis mismatch yield nonzero truncation error. The
	$\mathcal O(1/N)$ rate is an idealized asymptotic approximation, and
	the reduced-rank NMSE includes both $\epsilon_{\mathrm{sub}}$ and
	$\epsilon_{\mathrm{est}}$.
\end{remark}

\subsection{Oracle Reduced-Rank MMSE}
\label{subsec:oracle_reduced_mmse}

We next characterize $\epsilon_{\mathrm{est}}$ assuming oracle
hyperparameters.

\begin{proposition}[Oracle reduced-rank MMSE]
	\label{prop:oracle_reduced_mmse}
	Assume $\epsilon_{\mathrm{sub}} \approx 0$. Given $M$ observed ports
	with selection matrix $\mathbf S_{\mathcal O}$, the oracle
	reduced-rank MMSE is
	\begin{equation}
		\mathrm{MSE}_{\mathrm{oracle\text{-}rr}}
		= \sum_{k=1}^{r}
		\frac{p_k\,\sigma^2}{\sigma^2 + p_k\,\mu_k},
		\label{eq:oracle_rr_mmse}
	\end{equation}
	where $\mu_1 \ge \cdots \ge \mu_r \ge 0$ are the eigenvalues of
	$\mathbf G^{\mathsf H}\mathbf G \in \mathbb C^{r \times r}$ and
	$\{p_k\}$ are the oracle modal powers. The normalized MSE is
	\begin{equation}
		\mathrm{NMSE}_{\mathrm{oracle\text{-}rr}}
		= \frac{1}{\sum_{k=1}^{r} p_k}
		\sum_{k=1}^{r}
		\frac{p_k\,\sigma^2}{\sigma^2 + p_k\,\mu_k}.
		\label{eq:oracle_rr_nmse}
	\end{equation}
\end{proposition}

\begin{IEEEproof}
	Under $\mathbf y_t = \mathbf G\mathbf z_t + \mathbf n_t$, with
	$\mathbf z_t \sim \mathcal{CN}(\mathbf 0, \mathbf P)$ and
	$\mathbf n_t \sim \mathcal{CN}(\mathbf 0, \sigma^2\mathbf I_M)$,
	the posterior covariance is
	$\mathbf C_z = (\mathbf P^{-1}
	+ \sigma^{-2}\mathbf G^{\mathsf H}\mathbf G)^{-1}$.
	Since $\mathbf B^{\mathsf H}\mathbf B = \mathbf I_r$, the full-port
	MSE is $\operatorname{tr}(\mathbf C_z)$. For diagonal $\mathbf P$,
	this trace yields \eqref{eq:oracle_rr_mmse}.
\end{IEEEproof}

The oracle performance is governed by the {\em modal SNRs}
$\mathrm{SNR_m}=p_k / \sigma^2$ and the sensing eigenvalues $\{\mu_k\}$.

\begin{corollary}[Flat spectrum, ideal Clarke case]
	\label{cor:flat_spectrum}
	When $p_k = p$ for all $k$,
	\begin{equation}
		\mathrm{NMSE}_{\mathrm{oracle\text{-}rr}}
		= \frac{1}{r} \sum_{k=1}^{r}
		\frac{1}{1 + \mathrm{SNR_m} \cdot \mu_k},
		\label{eq:nmse_flat_spectrum}
	\end{equation}
	where the modal SNR equals to $p / \sigma^2$.
\end{corollary}

\begin{corollary}[High-SNR regime]
	\label{cor:high_snr}
	If $\mu_k > 0$ for all $k$, then as $\mathrm{SNR} \to \infty$,
	\begin{equation}
		\mathrm{NMSE}_{\mathrm{oracle\text{-}rr}}
		\to
		\frac{\sigma^2\,
			\operatorname{tr}\!\bigl(
			(\mathbf G^{\mathsf H}\mathbf G)^{-1}
			\bigr)}
		{\operatorname{tr}(\mathbf P)}.
		\label{eq:nmse_high_snr}
	\end{equation}
	If $\mu_k = \mu$ for all $k$, as for a tight frame, then
	$\mathrm{NMSE} \to 1 / (\mu \cdot \mathrm{SNR_m})$.
\end{corollary}

\begin{corollary}[Rank-deficient sensing, $M < r$]
	\label{cor:rank_deficient}
	When $M < r$, $\mathbf G \in \mathbb C^{M \times r}$ has at least
	$r - M$ zero eigenvalues. Therefore,
	\begin{equation}
		\mathrm{MSE}_{\mathrm{oracle\text{-}rr}}
		\ge \sum_{k:\,\mu_k = 0} p_k,
		\label{eq:mse_rank_deficient}
	\end{equation}
	which remains nonzero at arbitrarily high SNR. Thus, no estimator
	within the $r$-dimensional subspace can overcome the missing spatial
	dimensions.
\end{corollary}

\subsection{Phase Transition at \texorpdfstring{$M = r$}{M = r}}
\label{subsec:phase_transition}

The preceding corollaries imply a sharp transition at $M = r$.

\begin{proposition}[Phase transition]
	\label{prop:phase_transition}
	Consider a Clarke-type channel with $r = 2\lfloor W\rfloor + 1$
	active modes and $M$ selected ports.
	\begin{enumerate}
		\item[\emph{(i)}]
		\emph{Impossibility regime, $M < r$.}
		For any $\mathcal O$ with $|\mathcal O| = M < r$,
		\begin{equation}
			\mathrm{NMSE}_{\mathrm{oracle\text{-}rr}}
			\ge \frac{(r - M)\,p_{\min}}{r\,\bar p},
			\label{eq:impossibility_bound}
		\end{equation}
		where $p_{\min} = \min_k p_k$ and
		$\bar p = r^{-1}\sum_k p_k$. Under the flat Clarke model,
		$\mathrm{NMSE} \ge (r{-}M)/r$, independent of SNR.
		
		\item[\emph{(ii)}]
		\emph{Achievability regime, $M \ge r$.}
		There exist selections $\mathcal O$ with $|\mathcal O| = M \ge r$
		for which $\mathbf G$ has full column rank and
		\begin{equation}
			\mathrm{NMSE}_{\mathrm{oracle\text{-}rr}}
			\le \frac{r}{\mathrm{SNR_m} \cdot \mu_{\min}}
			\xrightarrow{\mathrm{SNR_m} \to \infty} 0,
			\label{eq:achievability_bound}
		\end{equation}
		where $\mu_{\min} > 0$ is the smallest eigenvalue of
		$\mathbf G^{\mathsf H}\mathbf G$.
	\end{enumerate}
\end{proposition}

\begin{IEEEproof}
	For Part~(i), the null space of $\mathbf G$ has dimension at least
	$r-M$, giving irreducible energy
	$\sum_{k:\mu_k=0} p_k \ge (r{-}M)\,p_{\min}$. Normalizing by
	$\operatorname{tr}(\mathbf P) = r\bar p$ gives
	\eqref{eq:impossibility_bound}. For Part~(ii), selecting $M \ge r$
	rows from the DFT sub-matrix
	$\mathbf B(\vartheta) \in \mathbb C^{N \times r}$ generically
	produces a full-column-rank matrix. Corollary~\ref{cor:high_snr} then
	applies, and $\sum_k \mu_k^{-1} \le r / \mu_{\min}$ yields
	\eqref{eq:achievability_bound}.
\end{IEEEproof}

Uniform selection provides a useful explicit characterization.

\begin{proposition}[NMSE under uniform port selection]
	\label{prop:uniform_selection}
	Suppose the $M \ge r$ observed ports are uniformly spaced over the
	$N$-port aperture, with
	$n_m = \lfloor (m{-}1)N/M \rfloor + 1$ for $m = 1,\dots,M$. In the
	large-$N$ regime, all eigenvalues of
	$\mathbf G^{\mathsf H}\mathbf G$ converge to $M/N$. Under the flat
	Clarke model,
	\begin{equation}
		\mathrm{NMSE}_{\mathrm{oracle\text{-}rr}}^{\mathrm{unif}}
		\approx \frac{1}{1 + \rho \cdot \mathrm{SNR_m}},
		\label{eq:nmse_uniform}
	\end{equation}
	where $\rho = M/N$ is the observation ratio.
\end{proposition}

\begin{IEEEproof}
	Uniform sub-sampling selects approximately every $(N/M)$th DFT row.
	Each active column $\mathbf f_k$ has selected energy
	$\|\mathbf S_{\mathcal O}\mathbf f_k\|^2 = M/N$. For $r < M$, the
	cross terms vanish without aliasing, giving
	$\mathbf G^{\mathsf H}\mathbf G \approx (M/N)\,\mathbf I_r$.
	Substitution into Corollary~\ref{cor:flat_spectrum} yields
	\eqref{eq:nmse_uniform}.
\end{IEEEproof}
\begin{figure}[t!]
	\centering
	\includegraphics[width=0.85\columnwidth]{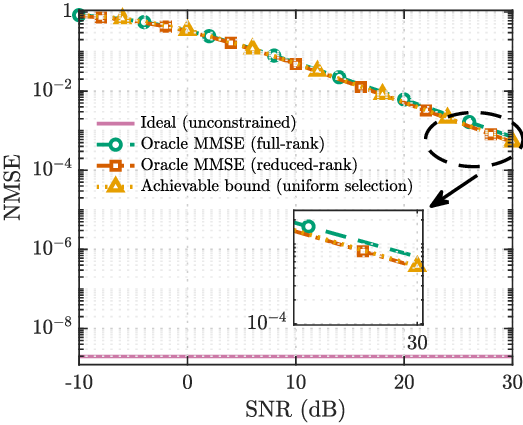}
	\caption{NMSE versus SNR with $M=10$ uniformly selected ports.
		System parameters, $N=64$, $W=2$, and $r=5$.}
	\label{fig:NMSE_SNR}
	\vspace{-2mm}
\end{figure}
\subsection{Learning Error and Snapshot Requirements}
\label{subsec:learning_error}

The learning term quantifies the cost of estimating, rather than
knowing, the hyperparameters.

\begin{proposition}[Snapshot requirements]
	\label{prop:min_T}
	Consider EM estimation of the $r$ modal powers $\{p_k\}$ and noise
	variance $\sigma^2$ from $T$ independent snapshots, each with $M$
	noisy observations.
	\begin{enumerate}
		\item[\emph{(i)}]
		\emph{Identifiability.}
		A necessary condition for identifying the $r + 1$ parameters from
		marginal statistics is
		\begin{equation}
			T \ge \bigl\lceil (r + 1) / M \bigr\rceil.
			\label{eq:T_necessary}
		\end{equation}
		Thus, one snapshot is sufficient when $M \ge r + 1$.
		
		\item[\emph{(ii)}]
		\emph{Estimation accuracy.}
		Under regularity conditions, the MLE of each $p_k$ is consistent,
		with
		\begin{equation}
			\operatorname{Var}(\widehat p_k)
			\approx \frac{p_k^2}{T} \cdot c_k(M, \sigma^2),
			\label{eq:pk_variance}
		\end{equation}
		where $c_k \ge 1$ depends on sensing geometry and SNR, and equals
		$1$ for directly observed modal coefficients. The corresponding
		normalized learning excess is
		\begin{equation}
			\frac{\epsilon_{\mathrm{learn}}}
			{\operatorname{tr}(\mathbf P)}
			\approx
			\frac{1}{T}
			\sum_{k=1}^{r}
			\frac{p_k\,c_k}{\operatorname{tr}(\mathbf P)}
			\cdot
			\frac{\mu_k^2\,\mathrm{SNR_m}^2}
			{(1 + p_k\mu_k/\sigma^2)^4}.
			\label{eq:learning_excess}
		\end{equation}
		
		Under the flat Clarke model with uniform selection, $p_k = p$,
		$\mu_k = M/N$, and $c_k = 1$. A sufficient condition for
		$\epsilon_{\mathrm{learn}} /
		\operatorname{tr}(\mathbf P) \le \varepsilon_{\mathrm{learn}}$ is
		\begin{equation}
			T \ge
			\frac{1}{\varepsilon_{\mathrm{learn}}}
			\cdot
			\frac{\rho^2\,\mathrm{SNR_m}^2}
			{(1 + \rho\,\mathrm{SNR_m})^4}.
			\label{eq:T_sufficient_flat}
		\end{equation}
	\end{enumerate}
\end{proposition}

\begin{IEEEproof}
	For Part~(i), each snapshot provides one realization of
	$\mathbf y_t \sim \mathcal{CN}(\mathbf 0, \mathbf R_t)$, whose
	second-order statistics depend on the $r + 1$ unknown parameters.
	The sample covariance spans the parameter space when
	$T \cdot \min(M, L_t) \ge r + 1$. For Part~(ii), the Fisher
	information for $p_k$ is
	\begin{equation}
		[\mathbf J]_{kk}
		= T\operatorname{tr}\!\Bigl(
		\mathbf R_t^{-1}
		\frac{\partial \mathbf R_t}{\partial p_k}
		\mathbf R_t^{-1}
		\frac{\partial \mathbf R_t}{\partial p_k}
		\Bigr),
		\label{eq:fim_pk}
	\end{equation}
	where
	$\partial \mathbf R_t / \partial p_k
	= \mathbf G\mathbf e_k\mathbf e_k^{\mathsf H}\mathbf G^{\mathsf H}$.
	Evaluating this expression and propagating the parameter uncertainty
	through the MMSE gives \eqref{eq:learning_excess}. Solving the
	resulting inequality yields \eqref{eq:T_sufficient_flat}.
\end{IEEEproof}

Equation~\eqref{eq:T_sufficient_flat} decreases with SNR. At high SNR,
each snapshot reveals the modal coefficients accurately and few samples
are needed to learn their second-order statistics. At low SNR, more
snapshots are required. For $T = 1$, learning is most challenging,
although $M \gg r$ can still provide substantial statistical
information.
\begin{figure}[t!]
	\centering
	\includegraphics[width=0.85\columnwidth]{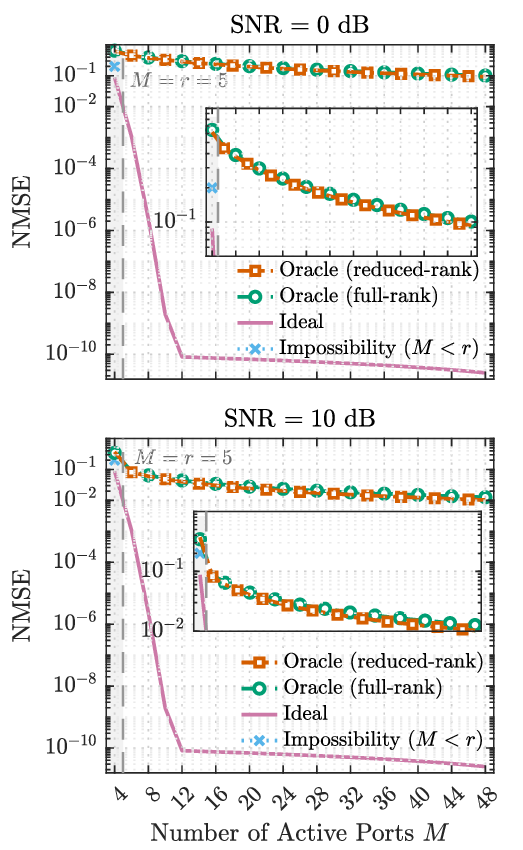}
	\caption{NMSE versus the number of observed ports $M$ at
		$\mathrm{SNR}=0$ dB and $10$ dB. The gray region denotes $M<r$,
		and the dashed line marks $M=r=5$. System parameters, $N=64$,
		$W=2$, and uniform port selection.}
	\label{fig:NMSE_M}
	\vspace{-2mm}
\end{figure}
\section{Numerical Results}\label{sec:numerical_results}

We consider a Clarke correlated FAS with $N=64$ candidate ports and
normalized aperture $W=2$, which gives the effective modal dimension
$r=2\lfloor W\rfloor+1=5$. The candidate ports are uniformly spaced,
and the observed ports are selected uniformly over the aperture. Since
$[\boldsymbol{\Sigma}]_{n,n}=1$, the per-port SNR is
$\mathrm{SNR}=1/\sigma^2$. The ideal bound denotes the unconstrained
rank-$M$ lower bound recently developed in \cite{Ideal_bound}, which represents the best performance achievable with full knowledge of the spatial covariance: $
	\mathrm{NMSE}_{\mathrm{ideal}}(M)
	= \frac{\sum_{k=M+1}^{N} \lambda_k}
	{\sum_{k=1}^{N} \lambda_k}.$
The full-rank oracle uses
the true covariance and noise variance. The reduced-rank oracle uses the
true modal basis, modal powers, and noise variance, and reports only
$\epsilon_{\mathrm{est}}/\operatorname{tr}(\boldsymbol{\Sigma}_{\mathrm{true}})$.
Therefore, it excludes the subspace truncation term
$\epsilon_{\mathrm{sub}}$.

Fig.~\ref{fig:NMSE_SNR} evaluates the NMSE for
$\mathrm{SNR}\in[-10,30]$ dB with $M=10$ active ports. Since $M>r$, the
modal sensing matrix has full column rank, and all port selection based
curves decrease continuously with SNR. In the high SNR regime, the
oracle curves exhibit an approximately inverse SNR decay, which agrees
with Corollary~\ref{cor:high_snr}. The achievable uniform selection
bound closely follows the reduced-rank oracle over the complete SNR
range. This confirms that uniform port selection approximately satisfies
$\mathbf G^{\mathsf H}\mathbf G\approx(M/N)\mathbf I_r$, as predicted by
Proposition~\ref{prop:uniform_selection}. The reduced rank curve can lie
slightly below the full-rank oracle because it contains only the
in-subspace estimation error and excludes $\epsilon_{\mathrm{sub}}$.

Fig.~\ref{fig:NMSE_M} evaluates the NMSE for
$M\in\{4,6,\ldots,48\}$ at $\mathrm{SNR}=0$ dB and $10$ dB. The gray
region represents $M<r$, and the vertical dashed line marks the critical
threshold $M=r=5$. When $M<r$, the effective sensing matrix is rank
deficient, and the missing modal dimensions produce a nonzero
irreducible error. Increasing SNR alone cannot remove this error, in
agreement with Corollary~\ref{cor:rank_deficient}. Once $M$ exceeds $r$,
both oracle curves decrease monotonically as more ports are observed,
and the improvement is more evident at the higher SNR. These results
verify the phase transition in Proposition~\ref{prop:phase_transition}.
They also confirm that the recovery performance above the threshold is
controlled by both the observation ratio $M/N$ and SNR. The rapid decay
of the ideal bound near $M=r$ further supports the effective modal
dimension $r\approx2W+1$.

\section{Conclusion}\label{sec:conclusion}
This paper investigated full-port CSI reconstruction for FASs without prior knowledge of the spatial covariance under the Clarke model. We showed that the channel admits a low-dimensional modal representation whose effective dimension is governed primarily by the physical aperture, revealing a sharp recoverability threshold at $M=r$. When $M<r$, reconstruction is fundamentally underdetermined, whereas for $M\geq r$, the NMSE decreases with increasing SNR and number of observed ports, even without spatial covariance information. These results characterize the fundamental role of the modal dimension in prior-free CSI reconstruction and provide useful guidance for determining the required number of observed ports. 

Future work will investigate whether the identified modal threshold remains valid in more realistic propagation environments, particularly under non-isotropic and nonstationary scattering, where the modal structure and effective channel dimension may deviate from those predicted by the idealized Clarke model.

\balance
\bibliographystyle{IEEEtran}
\bibliography{references}
\end{document}